\documentclass{article}
\usepackage{amsfonts,amssymb}
\usepackage{amsmath}
\usepackage{graphicx}

\newif\ifpics
\picstrue 

\usepackage{color} 
   \definecolor{cites}{rgb}{0.50 , 0.00 , 0.00}  
   \definecolor{urls} {rgb}{0.00 , 0.00 , 0.50}  
   \definecolor{links}{rgb}{0.00 , 0.00 , 0.50}   

\usepackage{hyperref} 
\hypersetup{  
      colorlinks=true,   
      citecolor=cites,   
      urlcolor=urls,     
      linkcolor=links,   
      pdfstartview=FitH,       
      bookmarksopen=false      
}

\newcommand\C{{\mathbb C}}
\newcommand\D{{\mathbb D}}

\newcommand\N{{\mathbb N}}

\newcommand\bT{{\bf T}}
\newcommand\R{{\mathbb R}}

\newcommand\Z{{\mathbb Z}}

\newcommand\eps{{\varepsilon}}
\newcommand\spec{{\rm spec}\,}  
\newcommand\specn{{\rm spec}}   
\newcommand\speps{{\rm spec}_\eps\,}
\newcommand\spessp{{\rm spec}_{\rm ess}^p}
\newcommand\spess{{\rm spec}_{\rm ess}\,}
\newcommand\spptp{{\rm spec}_{\rm point}^p\,}
\newcommand\sppt{{\rm spec}_{\rm point}^\infty\,}

\newcommand\clos{{\rm clos\,}}
\newcommand\conv{{\rm conv}}
\newcommand\ri{{\rm i}}

\newtheorem{theorem}{Theorem}[section]

\newtheorem{proposition}[theorem]{Proposition}

\newenvironment{remark}
 {\par\noindent\refstepcounter{theorem}{\bf Remark \thetheorem\,}}
 {\raisebox{1mm}{\framebox{}}\pagebreak[2]}

\newcommand\Proofend{\rule{2mm}{2mm}}

\newenvironment{proof}
 {\par\noindent{\bf Proof.}}
 {\Proofend\pagebreak[2]}

\newenvironment{proofof}[1]
 {\par\noindent{\bf Proof of #1.}}
 {\rule{2mm}{2mm}\pagebreak[2]}

\numberwithin{figure}{section}  

\begin{document}
\title{\bf Eigenvalue problem meets Sierpinski triangle:\\computing the spectrum of a non-self-adjoint random operator}
\author{{\sc Simon N. Chandler-Wilde},\quad {\sc Ratchanikorn Chonchaiya}\\[3mm] and\quad {\sc Marko Lindner}}
\date{\today}
\maketitle
\begin{quote}
\renewcommand{\baselinestretch}{1.0}
\footnotesize {\sc Abstract.} The purpose of this paper is to prove
that the spectrum of the non-self-adjoint one-particle Hamiltonian
proposed by J.~Feinberg and A.~Zee (Phys.\ Rev.\ E
59 (1999), 6433--6443) has interior points.
We do this by first recalling that the spectrum of this random operator is the union of the set of $\ell^\infty$ eigenvalues of all infinite matrices with the same structure. We then construct an infinite matrix of this structure
for which every point of the open unit disk is an $\ell^\infty$
eigenvalue, this following from the fact that the components of the eigenvector are polynomials in the spectral parameter whose non-zero coefficients are $\pm 1$'s, forming the pattern of an infinite discrete Sierpinski triangle.
\end{quote}


\section{Introduction and Notations} \label{sec_intro}
In this paper we study infinite matrices of the form
\begin{equation} \label{eq:Ab}
\ \left(\begin{array}{cccccc} \ddots&\ddots\\
\ddots&0&1\\
&b_{-1}&0&1\\
&&b_{0}&0&1\\
&&&b_{1}&0&\ddots\\
&&&&\ddots&\ddots
\end{array}\right)
\end{equation}
with $b_k\in\{\pm 1\}:=\{-1,+1\}$ for all $k\in\Z$. We think of
\eqref{eq:Ab} as a linear operator acting via matrix-vector
multiplication on $\ell^p(\Z)$, the standard space of bi-infinite
complex sequences with $p\in [1,\infty]$. By $\{\pm 1\}^\Z$ we
denote the set of all sequences $b=(b_k)_{k\in\Z}$ with $b_k\in\{\pm
1\}$ for all $k\in\Z$, and we refer to the operator on $\ell^p(\Z)$
that is induced by the matrix \eqref{eq:Ab} as $A^b$. For
convenience, we will also refer to the matrix \eqref{eq:Ab} as
$A^b$. For $p\in [1,\infty]$ and $b\in\{\pm 1\}^\Z$, we write
\begin{eqnarray*}
\specn^p A^b&:=&\{\lambda\in\C\,:\,A^b-\lambda I \textrm{ is not
invertible on } \ell^p(\Z)\},\\
\spessp A^b&:=&\{\lambda\in\C\,:\,A^b-\lambda I \textrm{ is not
Fredholm on } \ell^p(\Z)\},\\
\spptp A^b&:=&\{\lambda\in\C\,:\,A^b-\lambda I \textrm{ is not
injective on } \ell^p(\Z)\}.
\end{eqnarray*}
Because \eqref{eq:Ab} is a band matrix, it holds (see
\cite{Kurbatov} and \cite{Li:Wiener}) that $\specn^p A^b$ and
$\spessp A^b$ do not depend on $p\in [1,\infty]$, so that it makes sense to abbreviate these as $\spec A^b$ and $\spess A^b$ in what follows. Note
however that the set of eigenvalues, $\spptp A^b$, does depend on $p$.

\newpage 

Physicists have studied the operator $A^b$ as the
(non-self-adjoint) Hamiltonian of a particle hopping (asymmetrically)
on a 1-dimensional lattice
\cite{FeinZee99,FeinZee99a,CicConMol2000,HolzOrlandZee}. Applications of such and
related Hamiltonians, especially examples with random diagonals, include
vortex line pinning in superconductors and growth models in
population biology. The particular model \eqref{eq:Ab} was proposed
by Feinberg and Zee in \cite{FeinZee99}, and some properties of its
spectrum have been studied in \cite{CicConMol2000,HolzOrlandZee} (also see
Paragraph 37, in particular Figure 37.7c, in \cite{TrefEmb:Book}).

In all these studies the focus is on the case of a random sequence
$b\in\{\pm 1\}^\Z$. A related but completely deterministic concept
is that of a pseudo-ergodic sequence. In accordance with Davies
\cite{Davies2001:PseudoErg}, we call $b\in\{\pm 1\}^\Z$ {\sl
pseudo-ergodic} if every finite pattern of $\pm 1$'s can be found
somewhere (as a string of consecutive entries) in $b$. If $b$ is
pseudo-ergodic (which is almost surely the case if all $b_k$,
$k\in\Z$, are independent (or at least not fully correlated)
samples from a random variable with values in $\{\pm 1\}$ and
nonzero probability for both $+1$ and $-1$) then, as a consequence
of \cite{CWLi2008:FC} (also see
\cite{CW.Heng.ML:tridiag,CWLi:Memoir,Li:BiDiag,Li:Habil} and cf.\ \cite{Davies2001:PseudoErg}), it holds
that
\begin{equation} \label{eq:spec}
\spec A^b\ =\ \spess A^b\ =\ \bigcup_{c\in\{\pm 1\}^\Z} \spec A^c\
=\ \bigcup_{c\in\{\pm 1\}^\Z} \sppt A^c.
\end{equation}
The contribution of \cite{CWLi2008:FC} is the third ``='' sign in
(\ref{eq:spec}) that enables, or at least simplifies, the explicit
computation of the spectra of particular pseudo-ergodic operators in
\cite{CW.Heng.ML:tridiag,CWLi:Memoir,Li:BiDiag}. The first ``=''
sign in (\ref{eq:spec}) follows immediately from the second; the
second comes from the Fredholm theory of much more general operators
and is typically expressed in the language of so-called limit
operators \cite{RaRoSi1998,RaRoSi:Book,Li:Book,CWLi:Memoir}. (A
similar equality, often with the closure taken over the union of
spectra, can be found in the literature on spectral properties of
Schr\"odinger and more general Jacobi operators
\cite{PasturFigotin,CarmonaLacroix,Davies2001:SpecNSA,Davies2001:PseudoErg,GoldKoru,Mantoiu1,Mantoiu,GeorgescuGolenia,
GeorgescuIfti2000,GeorgescuIfti2002,GeorgescuIfti2005,
RaRo_ess_spec_Schr_latt,LastSimon06,LastSimon99,Remling1,Remling2}.
The three last papers also shed some light on the role of limit
operators in the study of the absolutely continuous spectrum.)

Note that, by \eqref{eq:spec}, the spectrum of $A^b$ does not depend
on the actual sequence $b$ -- as long as it is pseudo-ergodic. In
\cite{CW.Heng.ML:tridiag} we obtain information about the spectrum,
pseudospectrum and numerical range of the bi-infinite matrix
operator $A^b$, its contraction $A^b_+$ to the positive half axis
(a semi-infinite matrix) and the finite sections $A^b_n$ which, for $n\in\N$, are
$n\times n$ submatrices of \eqref{eq:Ab}. Explicitly and precisely, these related matrices are
\[
A^b_+\ =\ \left(\begin{array}{ccccc}
0&1\\
b_{1}&0&1\\
&b_{2}&0&1\\
&&b_{3}&0&\ddots\\
&&&\ddots&\ddots
\end{array}\right)
\qquad\textrm{and}\qquad A^b_n\ =\ \left(\begin{array}{ccccc}
0&1\\
b_{1}&0&1\\
&b_{2}&0&\ddots\\
&&\ddots&\ddots&1\\
&&&b_{n-1}&0
\end{array}\right),
\]
where in the case $n=1$ we set $A^b_1 = (0)$. We explore in some detail in \cite{CW.Heng.ML:tridiag} the interrelations between the spectra and
pseudospectra of $A^b$, $A^b_+$ and $A^b_n$. Here, for $\eps>0$ and a bounded operator $A$ on $\ell^2(\N)$ or $\ell^2(\Z)$, or on $\C^n$ equipped with the 2-norm, we define the {\sl $\eps$-pseudospectrum} of $A$ (see e.g.
\cite{BoeLi:Pseudospectrum,TrefEmb:Book}) by
\[
\speps A\ :=\ \{\lambda\in\C\,:\, \lambda\in\spec A\textrm{ or }
\|(A-\lambda I)^{-1}\|>1/\eps\},
\]
 where $\|\cdot\|$
is the induced operator norm. It is convenient also to use the notation $\specn_0 A:=\spec A$.
Note that the finite matrix $A^b_n$ only depends on the $n-1$ values
$b_1,...,b_{n-1}\in\{\pm 1\}$. Recognising this, we will use  the notation $A^{b'}_n$, where $b'=(b_1,...,b_{n-1})\in
\{\pm 1\}^{n-1}$, as an alternative notation for the same matrix $A^b_n$.

Here is a summary of our results from \cite{CW.Heng.ML:tridiag}:

\begin{theorem} \label{thm:tridiag} \cite{CW.Heng.ML:tridiag}
If $b\in\{\pm 1\}^\Z$ is pseudo-ergodic (which holds almost surely if
$b$ is random in the sense discussed above) then the following statements hold.
\begin{itemize}
\item[\bf a) ] $\spec A^b$ is invariant under reflection about either axis as well
as under a $90^o$ rotation around the origin.

\item[\bf b) ] Provided the ``positive'' part of the sequence $b$ (by which we mean $(b_k)_{k\in\N}$) is itself pseudo-ergodic (contains every finite pattern of $\pm1$'s), then, for all $\eps\ge 0$ one has
\[
\speps A^b\ =\ \speps A^b_+.
\]

\item[\bf c) ] The numerical range of $A^b$ (considered as an operator on $\ell^2(\Z)$) is
$$
W(A^b)\ =\ \{x+\ri y\,:\,x,y\in\R,\, |x|+|y|< 2\},
$$
and $\spec\, A^b$ is a strict subset of the closure, $\clos(W(A^b))$, of the numerical range, so that
\[
\spec\, A^b\ \subsetneqq\ \{x+\ri y\,:\,x,y\in\R,\,|x|+|y|\leq 2\}.
\]

\item[\bf d) ] For every $n\in\N$, where $\Pi_n :=\left\{c\in \{\pm1\}^\Z: c \mbox{ is $n$-periodic}\right\}$, the set
\begin{equation} \label{eq:spperntri}
\pi_n\ :=\ \bigcup_{c\in\Pi_n}\spec A^c\ =\ \bigcup_{c\in\Pi_n}\sppt
A^c
\end{equation}
is contained in $\spec A^b$, by \eqref{eq:spec}. Each set $\pi_n$
consists of $k$ analytic arcs (see Figure \ref{fig:30pics1}) with
$\frac{2^n}{n}\le k\le 2^n$ that can be computed explicitly (as
unions of sets of eigenvalues of $n\times n$ matrices). In
particular,
$$
\pi_1\, =\, [-2,2] \cup [-2\ri, 2\ri] \quad \mbox{ and } \quad \pi_2
\,=\, \pi_1 \cup \{x+\ri y: -1\leq x \leq 1, \, y = \pm x\}.
$$

\item[\bf e) ] For all $n\in\N$ and  $\eps\ge 0$,
the set
\begin{equation} \label{eq:spfs}
\sigma_{n,\eps}\ :=\ \bigcup_{c\in\{\pm 1\}^{n-1}}\speps A^c_n
\end{equation}
is contained in $\speps A^b$ (see Figure \ref{fig:30pics2} for
$\eps=0$).

\item[\bf f) ] In the case of spectra ($\eps=0$),
the finite matrix eigenvalues $\sigma_n := \sigma_{n,0}$ from
\eqref{eq:spfs} are connected with the periodic operator spectra
$\pi_n$ from \eqref{eq:spperntri} by
\begin{equation} \label{eq:spefsinper}
\sigma_n\ \subset\ \pi_{2n+2}\ \subset\ \spec A^b
\end{equation}
for all $n\in\N$ (see Figure \ref{fig:pic5in12}).

\item[\bf g) ] As a special case of a much more general spectral inclusion
result from \cite{CW.Heng.ML:UpperBounds}, we can complement the
inclusion $\sigma_{n,\eps}\subset\speps A^b$ from e) by
\[
\sigma_{n,\eps}\ \subset\ \speps A^b\ \subset\
\sigma_{n,\,\eps+\eps_{n}}\qquad\textrm{and}\qquad \sigma_{n}\
\subset\ \spec A^b\ \subset\
\clos\!\!\left(\sigma_{n,\,\eps_{n}}\right),
\]
for $n\in\N$ and $\eps>0$, where $\eps_{n} = 4\sin \theta_n < 2\pi/(n+1)$, with $\theta_n$ the unique solution in the interval $\displaystyle{\left(\frac{\pi}{2(n+3)},\frac{\pi}{2(n+1)}\right)}$ of the equation
$$
2\cos\left((n+1)\theta\right)\ =\ \cos\left((n-1)\theta\right).
$$
\end{itemize}
\end{theorem}

\begin{remark}
Note that the right ``='' sign in \eqref{eq:spperntri} holds because
$\spec A^c=\sppt A^c$ for all periodic sequences $c$, whereas the
right ``='' sign in \eqref{eq:spec} only holds as stated, with the
union taken over all $c\in \{\pm1\}^\Z$; the spectrum and point
spectrum of $A^c$ are different, in general, for specific $c\in\{\pm
1\}^\Z$.
\end{remark}

\begin{remark}
The inclusions in g) imply that $\spec A^b \subset
\clos(\sigma_{n,\varepsilon_n}) \subset \clos(\specn_{\varepsilon_n}
A^b)$. Since $\varepsilon_n\to 0$ so that
$\clos(\specn_{\varepsilon_n} A^b)\to \spec A^b$ in the Hausdorff
metric \cite{TrefEmb:Book} as $n\to\infty$, it follows that
$\clos(\sigma_{n,\varepsilon_n})\to \spec A^b$ as $n\to\infty$. For
small values of $n$ the upper bound
$\clos(\sigma_{n,\varepsilon_n})$ can be evaluated very explicitly.
In particular, $\theta_1=\pi/6$ so that $\varepsilon_1 = 2$ and,
since $A_1^c = (0)$, we obtain that $\spec A^b \subset
\clos(\sigma_{1,\varepsilon_1}) = \{\lambda \in \C:|\lambda|\leq
2\}$. The result in c) above, that $\spec A^b$ is a strict subset of
the closure of the numerical range, comes from the bound in g)
applied with $n=34$, when $\sigma_{n,\varepsilon_n}$ is the union of
the pseudospectra of $2^{33} \approx 8.6\times 10^9$ matrices of
size $34\times 34$.
\end{remark}

\noindent
\ifpics{  
\begin{center}
\begin{tabular}{ccccc}
\includegraphics[width=0.175\textwidth]{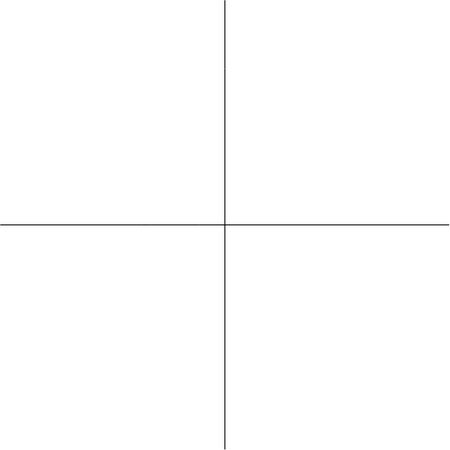}&
\includegraphics[width=0.175\textwidth]{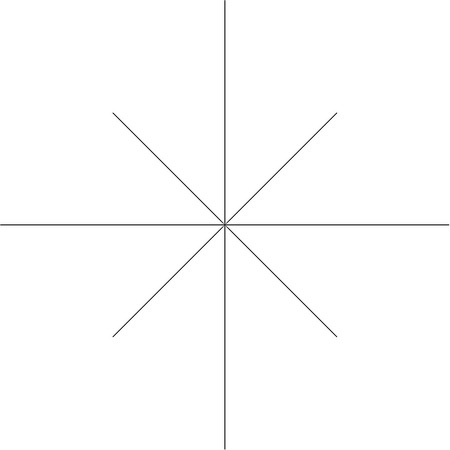}&
\includegraphics[width=0.175\textwidth]{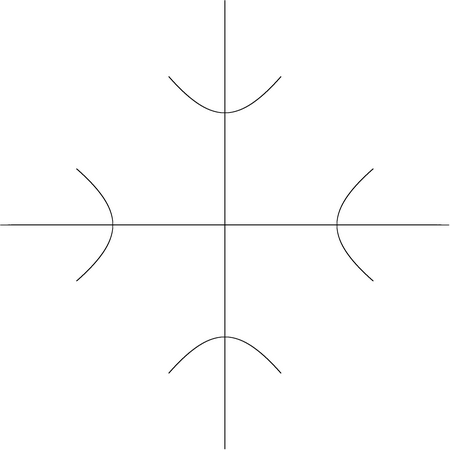}&
\includegraphics[width=0.175\textwidth]{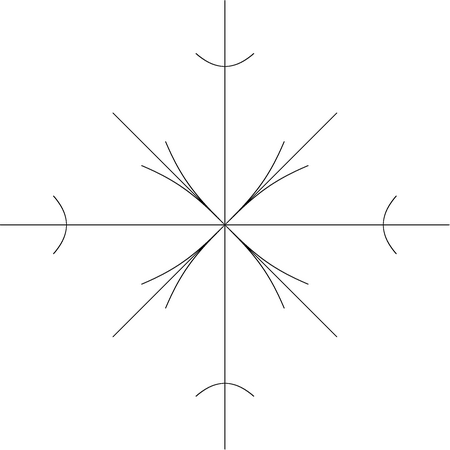}&
\includegraphics[width=0.175\textwidth]{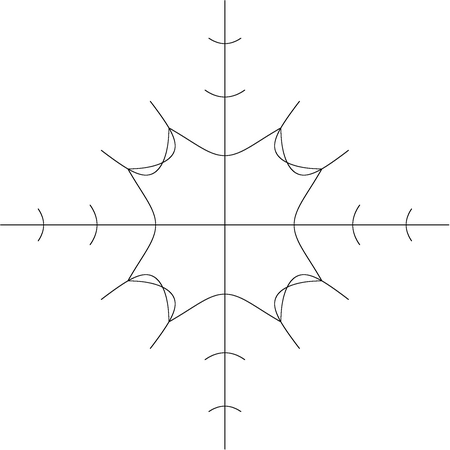}\\
\includegraphics[width=0.175\textwidth]{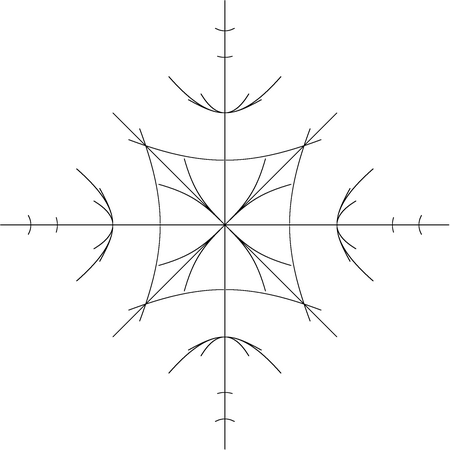}&
\includegraphics[width=0.175\textwidth]{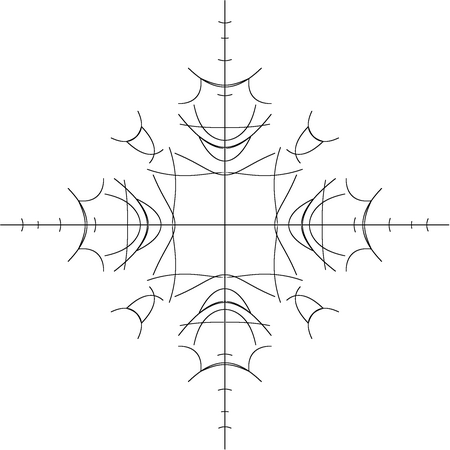}&
\includegraphics[width=0.175\textwidth]{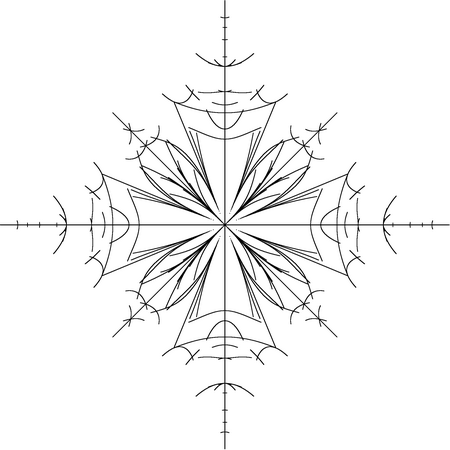}&
\includegraphics[width=0.175\textwidth]{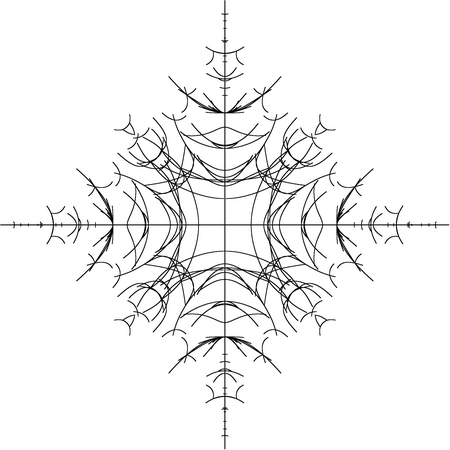}&
\includegraphics[width=0.175\textwidth]{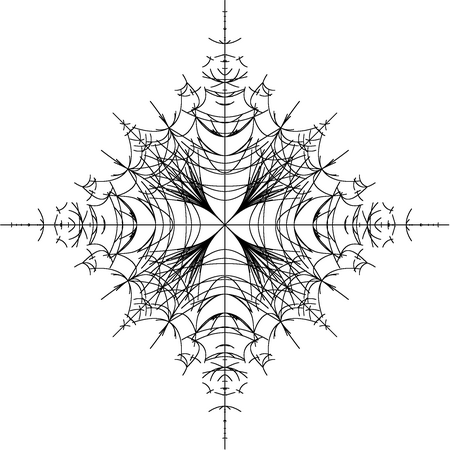}\\
\includegraphics[width=0.175\textwidth]{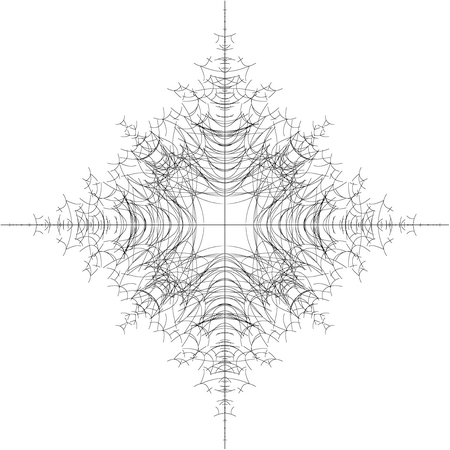}&
\includegraphics[width=0.175\textwidth]{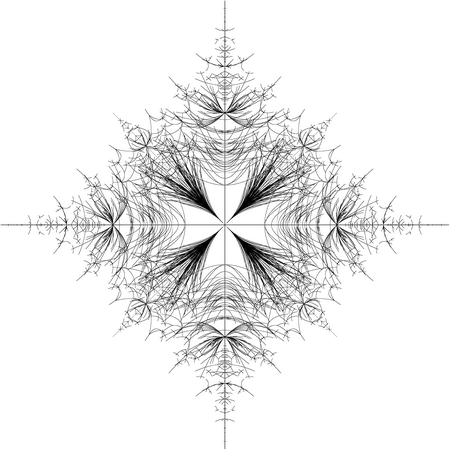}&
\includegraphics[width=0.175\textwidth]{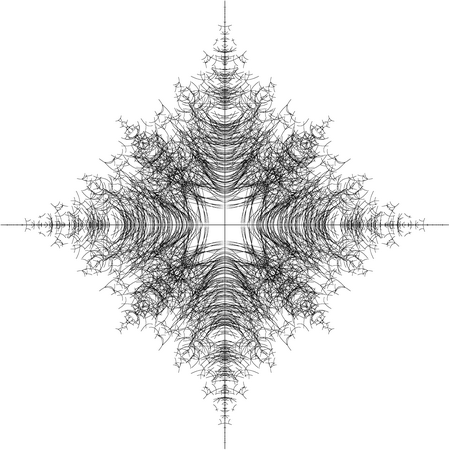}&
\includegraphics[width=0.175\textwidth]{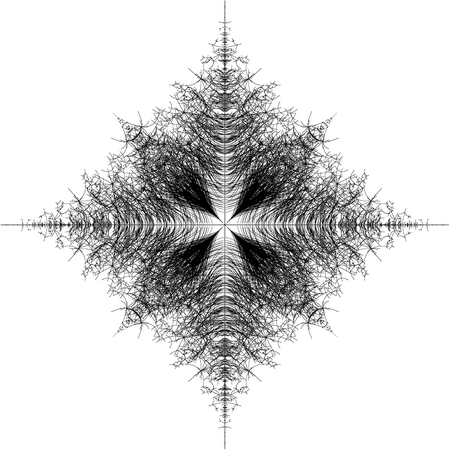}&
\includegraphics[width=0.175\textwidth]{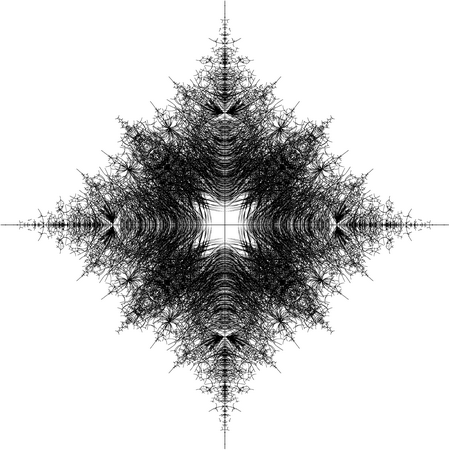}\\
\includegraphics[width=0.175\textwidth]{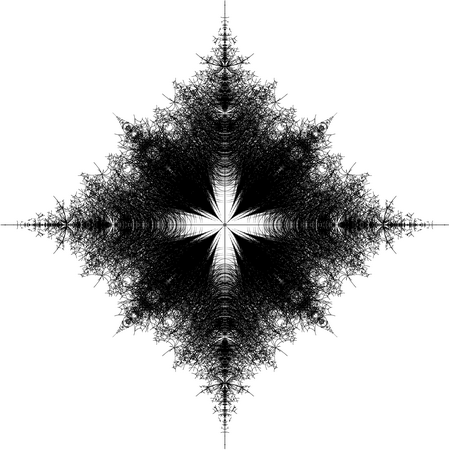}&
\includegraphics[width=0.175\textwidth]{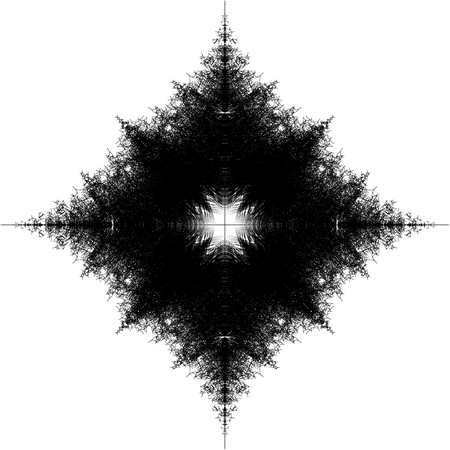}&
\includegraphics[width=0.175\textwidth]{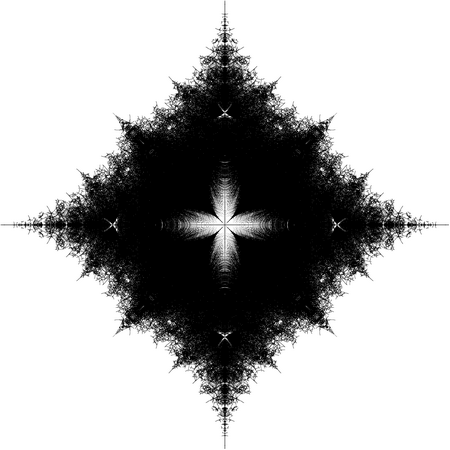}&
\includegraphics[width=0.175\textwidth]{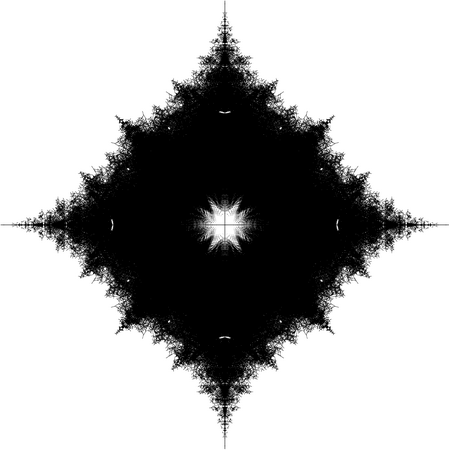}&
\includegraphics[width=0.175\textwidth]{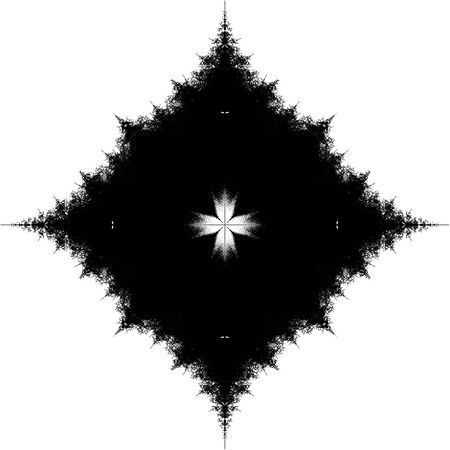}\\
\includegraphics[width=0.175\textwidth]{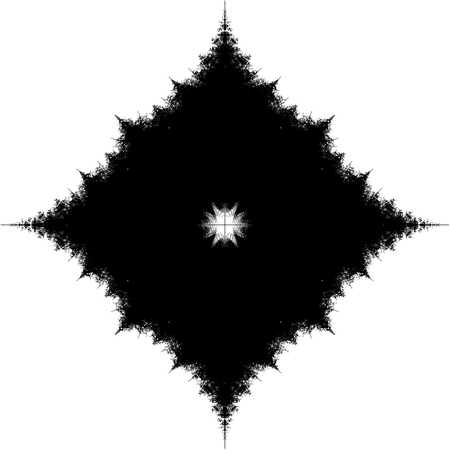}&
\includegraphics[width=0.175\textwidth]{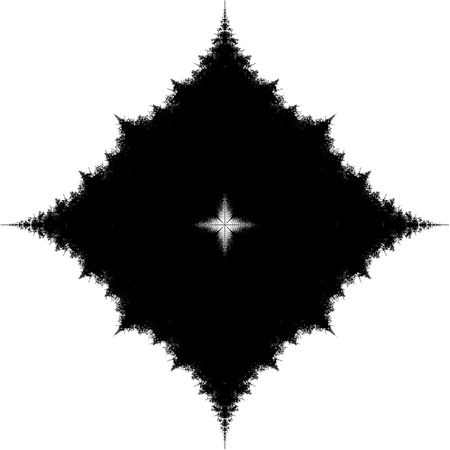}&
\includegraphics[width=0.175\textwidth]{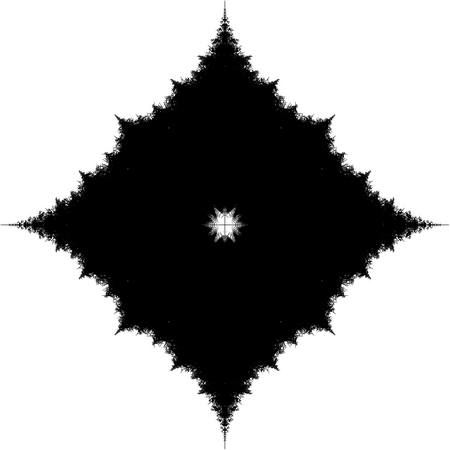}&
\includegraphics[width=0.175\textwidth]{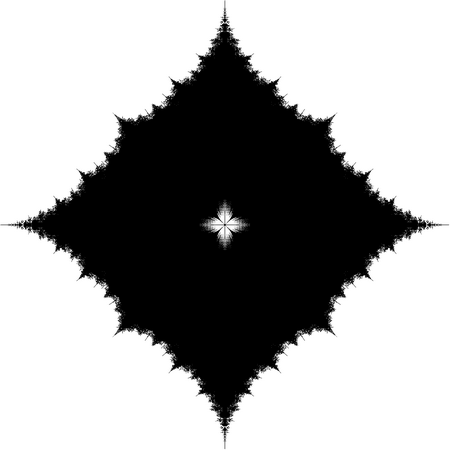}&
\includegraphics[width=0.175\textwidth]{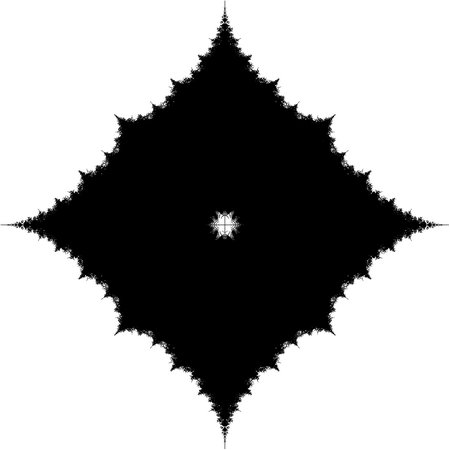}\\
\includegraphics[width=0.175\textwidth]{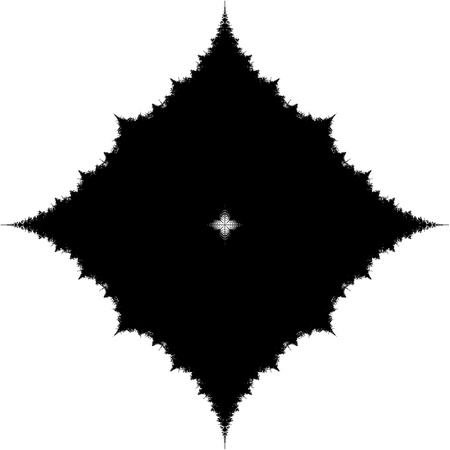}&
\includegraphics[width=0.175\textwidth]{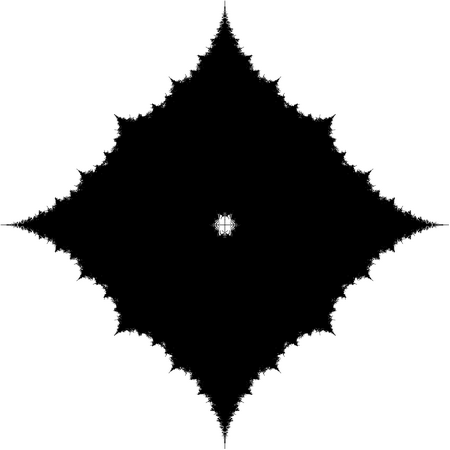}&
\includegraphics[width=0.175\textwidth]{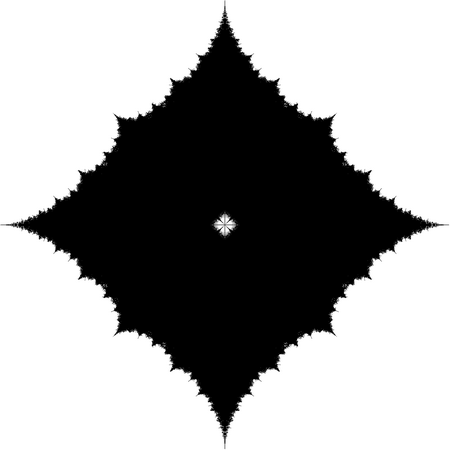}&
\includegraphics[width=0.175\textwidth]{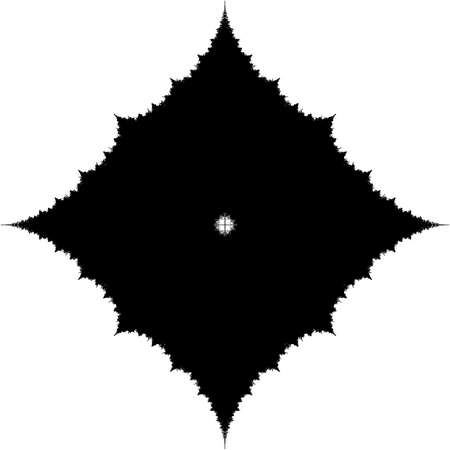}&
\includegraphics[width=0.175\textwidth]{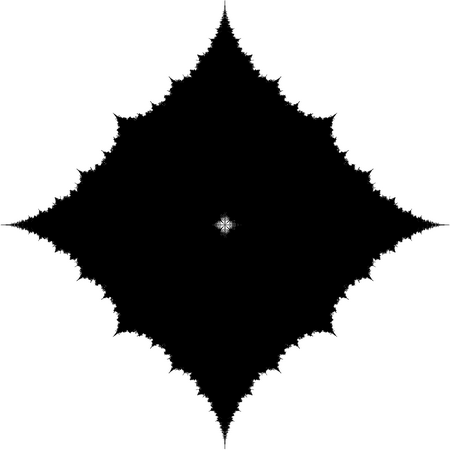}
\end{tabular}
\end{center}
}\fi 
\begin{figure}[h]
\caption{\footnotesize Our figure shows the sets $\pi_n$, as defined
in \eqref{eq:spperntri}, for $n=1,...,30$. Recall that $\pi_1 = [-2,2]\cup [-2\ri,2\ri]$ and that, for each $n$, $\pi_1\subset \pi_n\subset \{x+\ri y: x,y\in \R, \, |x|+|y| \leq 2\}$. Note also that spectra of
periodic infinite matrices can be expressed analytically (by Fourier
transform techniques, see e.g. \cite{BoeSi,Davies2007:Book}) and
that each set $\pi_n$ consists of $k$ analytic arcs, where $2^n/n\le
k\le 2^n$.} \label{fig:30pics1}
\end{figure}

It follows from Theorem \ref{thm:tridiag} d) and e) that both
\[
\sigma_\infty\ :=\ \bigcup_{n=1}^\infty
\sigma_n\qquad\textrm{and}\qquad \pi_\infty\ :=\
\bigcup_{n=1}^\infty\pi_n
\]
are subsets of $\spec A^b$ (with $\sigma_\infty\subset \pi_\infty$,
by \eqref{eq:spefsinper}). These subsets consist of countably many
points and countably many analytic arcs, respectively, and so both
have zero (two-dimensional) Lebesgue measure. Indeed, it is not
clear from any of the results in Theorem \ref{thm:tridiag} (or other
results in the literature) whether $\spec A^b$ has positive Lebesgue
measure, in particular whether it has interior points. Related to
this question, Holz et al.\ \cite[Sections I, V, VI]{HolzOrlandZee},
conjecture that $\clos(\sigma_\infty)\subset \spec A^b$ has a
fractal dimension in the range $(1,2)$, and so has zero Lebesgue
measure.

The purpose of the current paper is to shed light on these questions
by constructing a sequence $c\in\{\pm 1\}^\Z$ for which $\sppt A^c$
contains the open unit disk. As a consequence of formula
\eqref{eq:spec} and the closedness of spectra, this shows that
$\spec A^b$ contains the closed unit disk and therefore has
dimension 2 and a positive Lebesgue measure. This is the main result
of the next section. Intriguingly we will see that the sequence $c$
constructed, while rather irregular, is such that each $\lambda$ in
the unit disk is an eigenvalue of $A^c$ with an eigenvector
$u\in\ell^\infty(\Z)$ whose components are polynomials in $\lambda$
with coefficients forming the regular self-similar pattern of a
discrete Sierpinski triangle \eqref{eq:table}.

We will finish the paper with our own conjecture on the geometry of
$\clos(\sigma_\infty)$ and $\spec A^b$.

\noindent
\ifpics{ 
\begin{center}
\begin{tabular}{ccccc}
\includegraphics[width=0.175\textwidth]{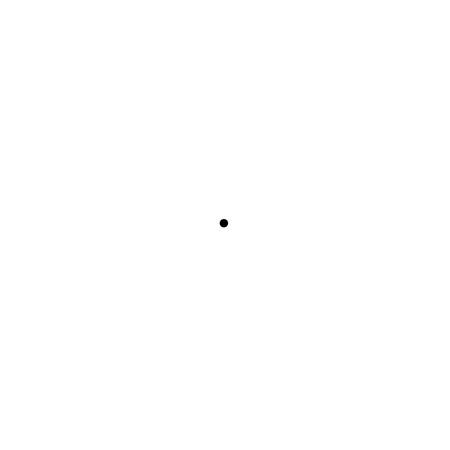}&
\includegraphics[width=0.175\textwidth]{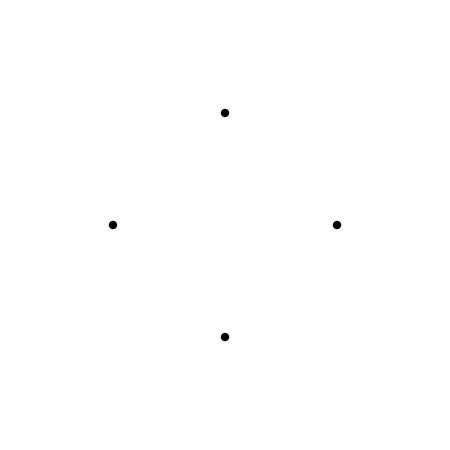}&
\includegraphics[width=0.175\textwidth]{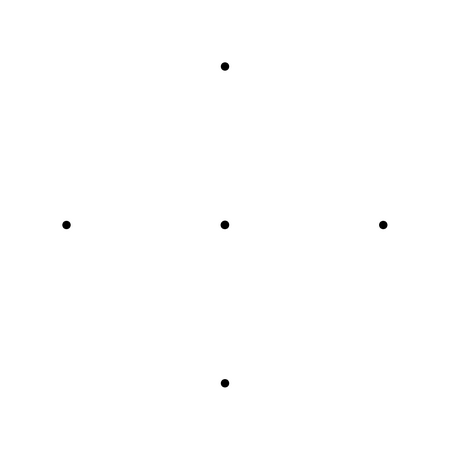}&
\includegraphics[width=0.175\textwidth]{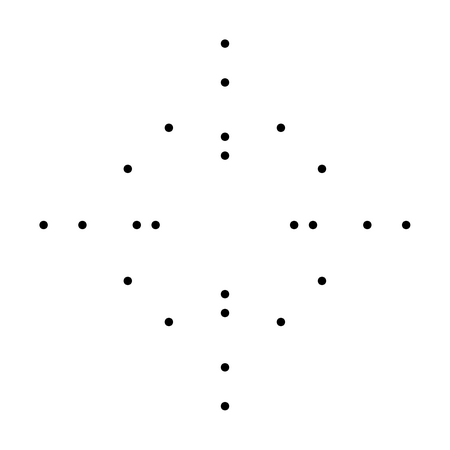}&
\includegraphics[width=0.175\textwidth]{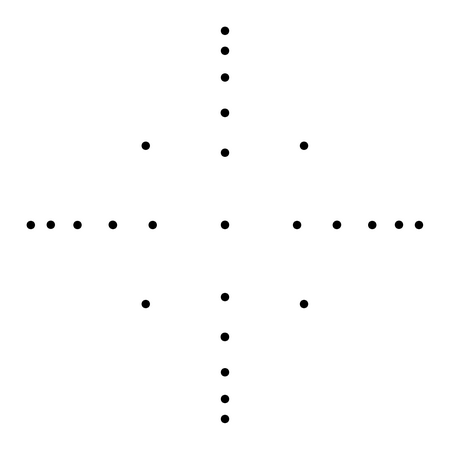}\\
\includegraphics[width=0.175\textwidth]{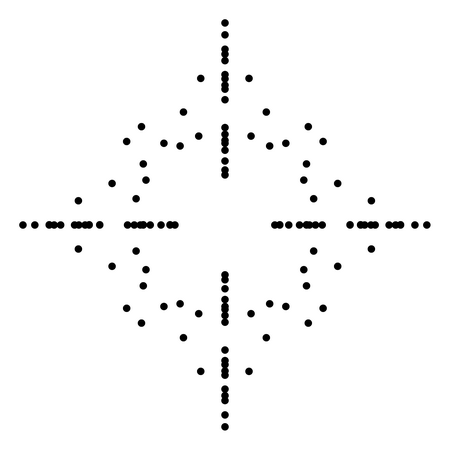}&
\includegraphics[width=0.175\textwidth]{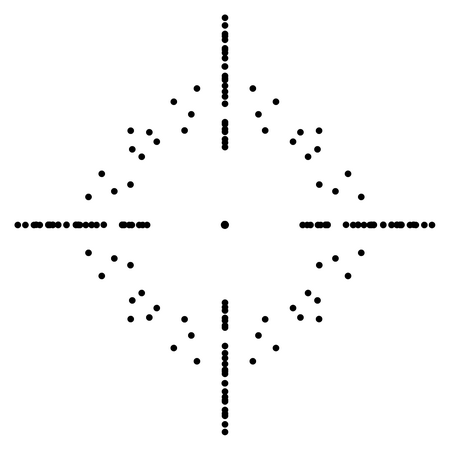}&
\includegraphics[width=0.175\textwidth]{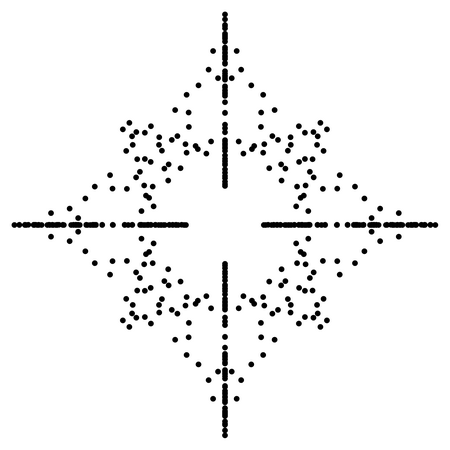}&
\includegraphics[width=0.175\textwidth]{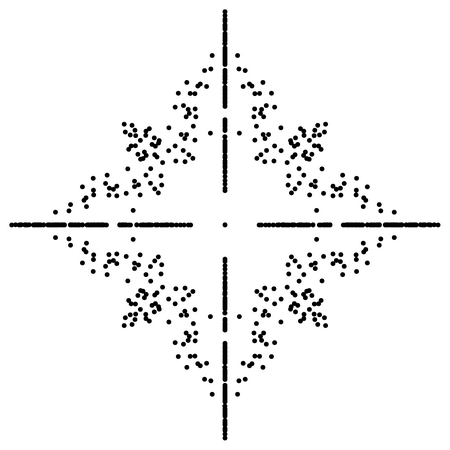}&
\includegraphics[width=0.175\textwidth]{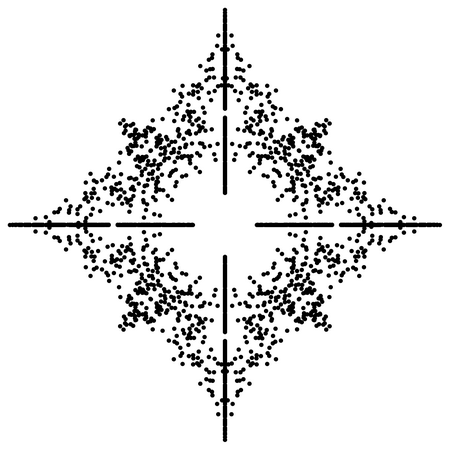}\\
\includegraphics[width=0.175\textwidth]{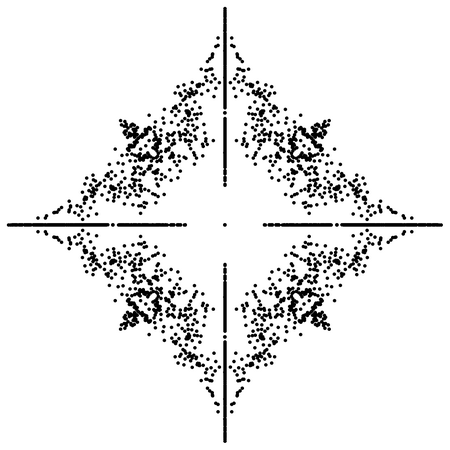}&
\includegraphics[width=0.175\textwidth]{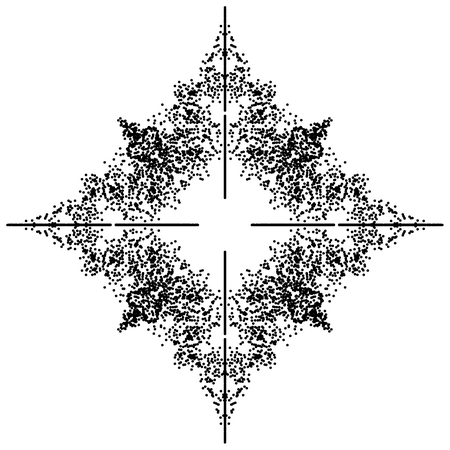}&
\includegraphics[width=0.175\textwidth]{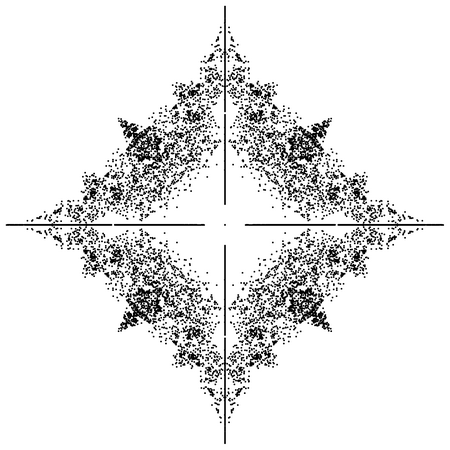}&
\includegraphics[width=0.175\textwidth]{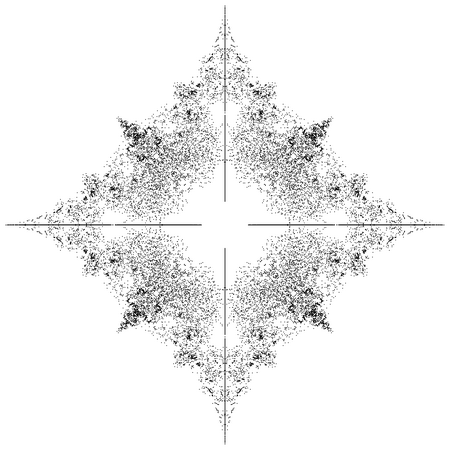}&
\includegraphics[width=0.175\textwidth]{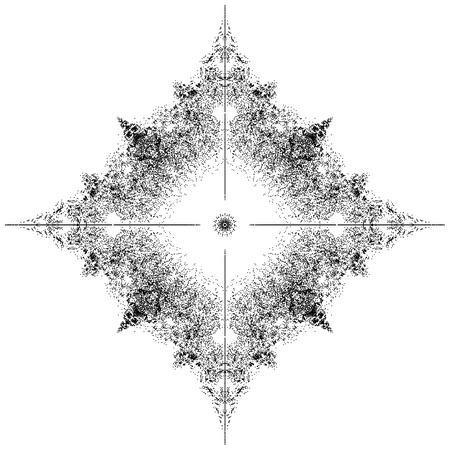}\\
\includegraphics[width=0.175\textwidth]{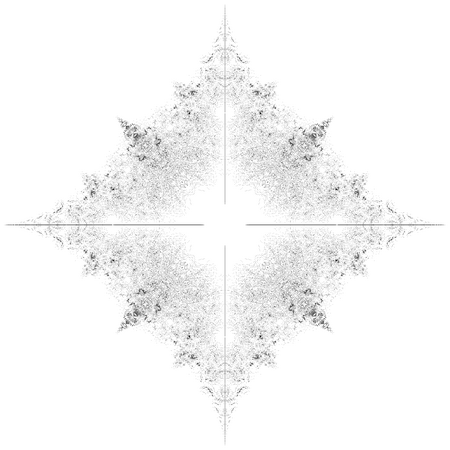}&
\includegraphics[width=0.175\textwidth]{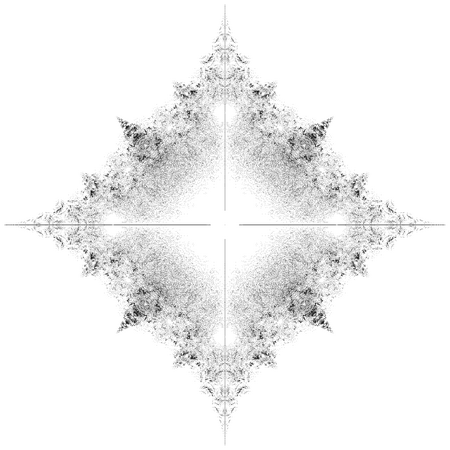}&
\includegraphics[width=0.175\textwidth]{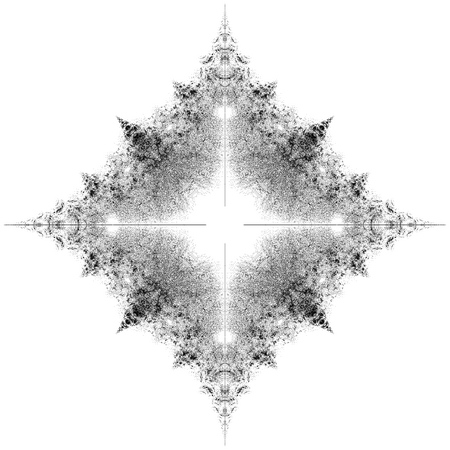}&
\includegraphics[width=0.175\textwidth]{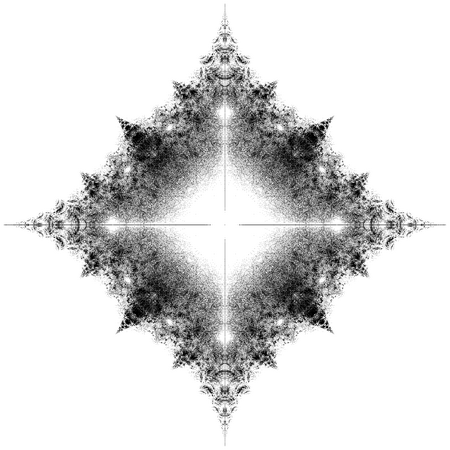}&
\includegraphics[width=0.175\textwidth]{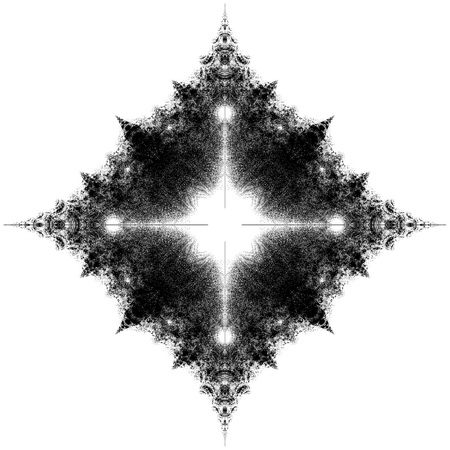}\\
\includegraphics[width=0.175\textwidth]{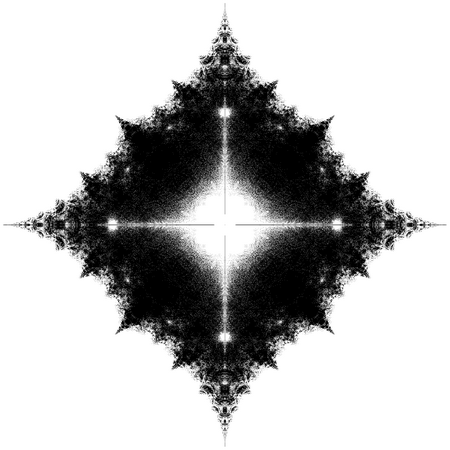}&
\includegraphics[width=0.175\textwidth]{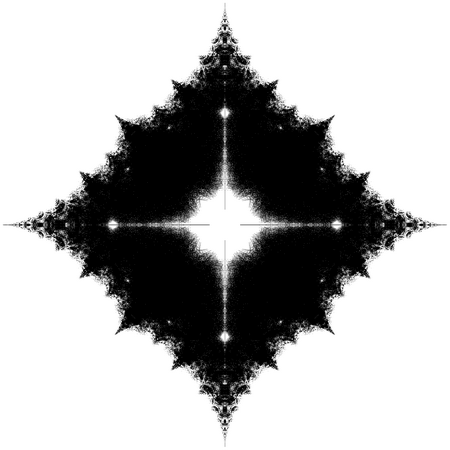}&
\includegraphics[width=0.175\textwidth]{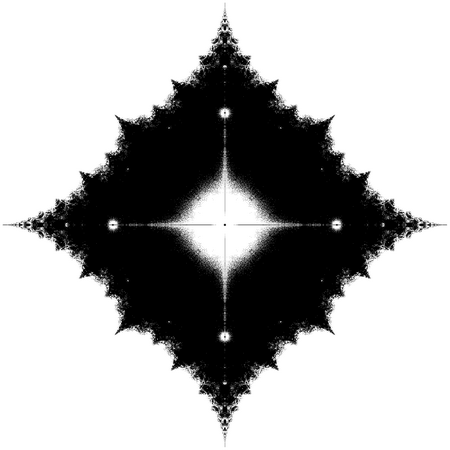}&
\includegraphics[width=0.175\textwidth]{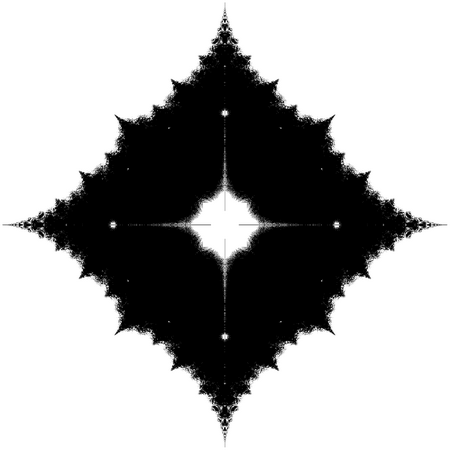}&
\includegraphics[width=0.175\textwidth]{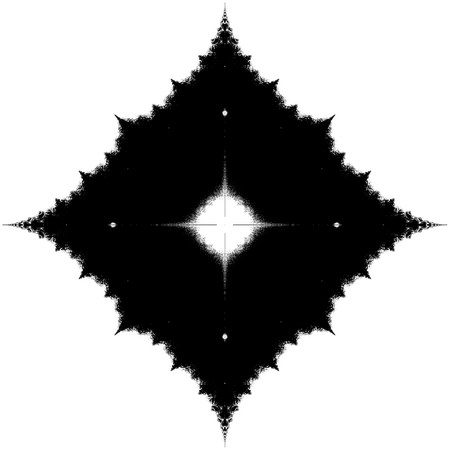}\\
\includegraphics[width=0.175\textwidth]{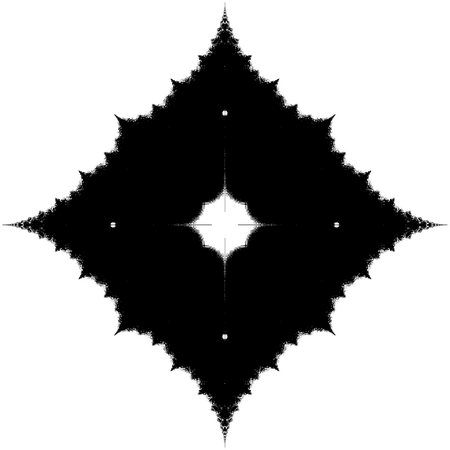}&
\includegraphics[width=0.175\textwidth]{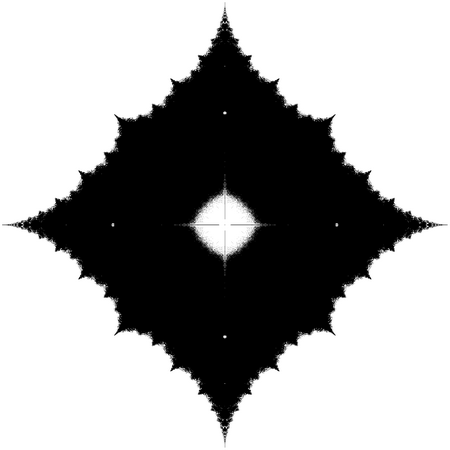}&
\includegraphics[width=0.175\textwidth]{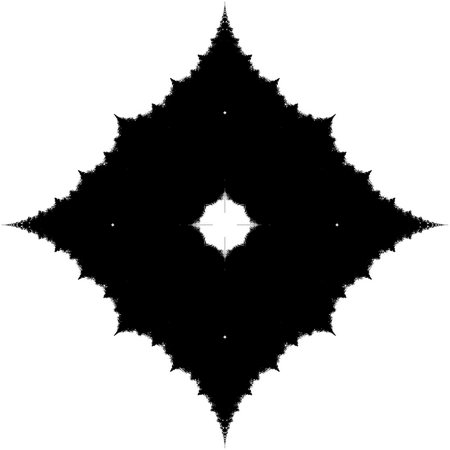}&
\includegraphics[width=0.175\textwidth]{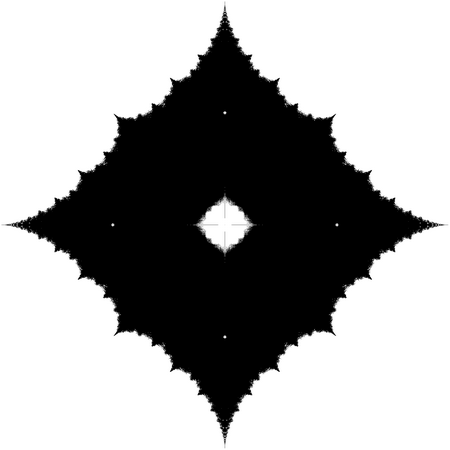}&
\includegraphics[width=0.175\textwidth]{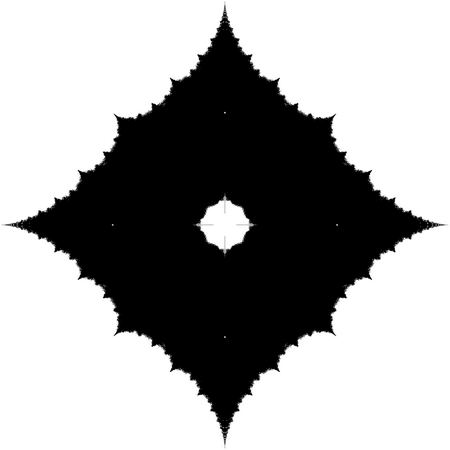}
\end{tabular}
\end{center}
}\fi 
\begin{figure}[h]
\caption{\footnotesize Our figure shows the sets $\sigma_n :=
\sigma_{n,0}$ of all $n\times n$ matrix eigenvalues, as defined in
\eqref{eq:spfs}, for $n=1,...,30$. Note that in the first pictures
(with only a few eigenvalues), we have used heavier pixels for the
sake of visibility. By \eqref{eq:spefsinper}, each of the sets with
$n=1,2,...,14$ in this figure is contained, respectively, in the set
number $2n+2$ of Figure \ref{fig:30pics1}.} \label{fig:30pics2}
\end{figure}

\noindent
\begin{center}
\ifpics{\includegraphics[width=\textwidth]{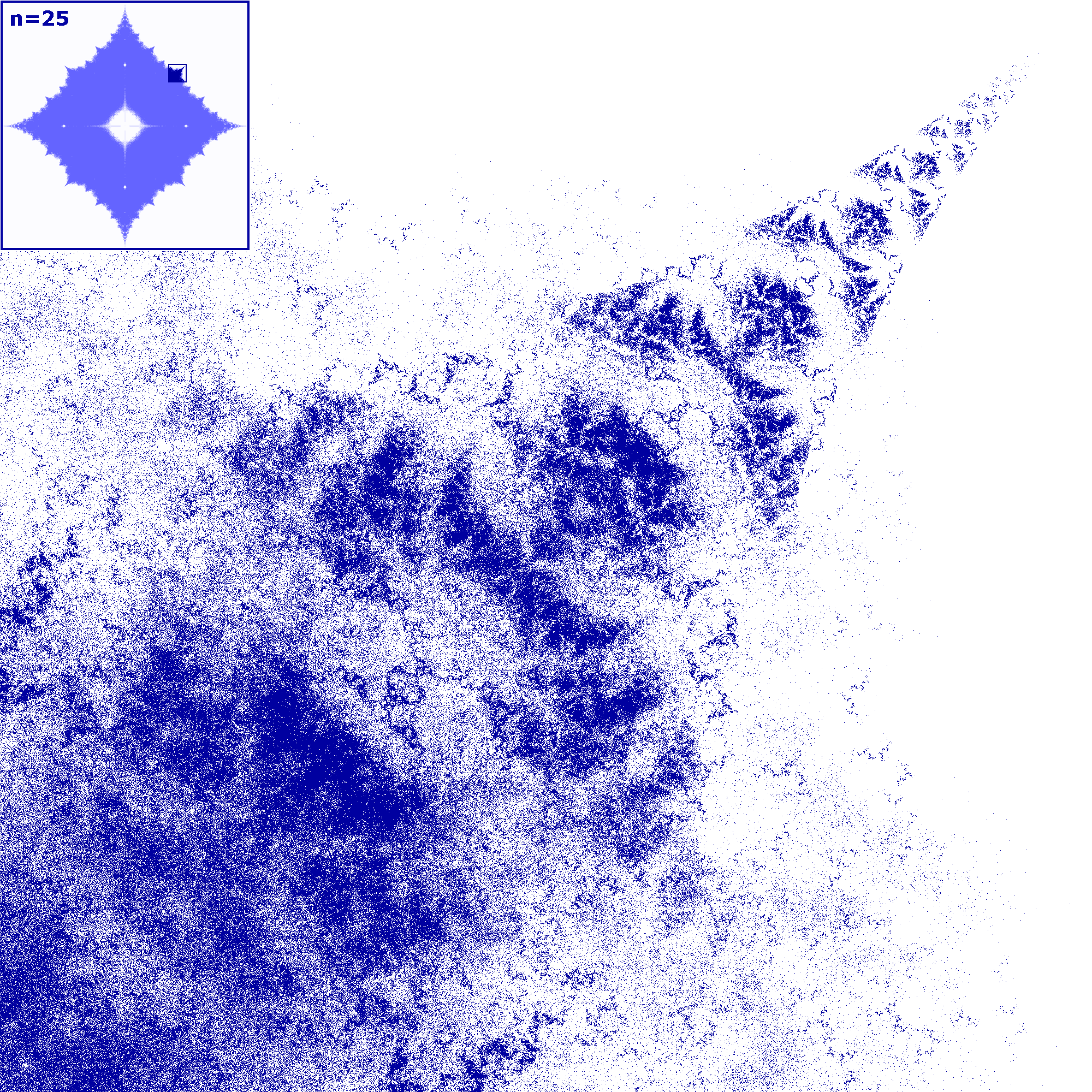}}\fi
\end{center}
\begin{figure}[h]
\caption{\footnotesize This is a zoom into $\sigma_{25}$ -- the 25th
picture of Figure \ref{fig:30pics2}. The location of this zoom is
near the point $1+i$, which is the midpoint of the northeast edge of
the square $\clos(W(A^b)) = \conv\{2,-2,2\ri,-2\ri\}$. The
picture clearly suggests self-similar features of the set $\sigma_{25}$.}
\label{fig:pic25zoom}
\end{figure}


\section{A sequence $c$ for which $\spec A^c$ contains the unit disk}
The formula \eqref{eq:spec} for the spectrum of $A^b$ when
 $b\in\{\pm 1\}^\Z$ is pseudo-ergodic motivates the following approach to
decide whether a given point $\lambda\in\C$ is in $\spec A^b$ or
not: look for a sequence $c\in\{\pm 1\}^\Z$ such that
$\lambda\in\sppt A^c$, in other words, such that there exists a non-zero $u\in\ell^\infty(\Z)$
with $A^cu=\lambda u$, i.e.
\begin{equation} \label{eq:rec}
u_{i+1}\ =\ \lambda u_{i}\ -\ c_{i}\,u_{i-1}
\end{equation}
for $i\in\Z$. If such a sequence $c$ exists then
$\lambda\in\spec A^b$ -- if not, then not.

\begin{center}
\ifpics{\includegraphics[width=\textwidth]{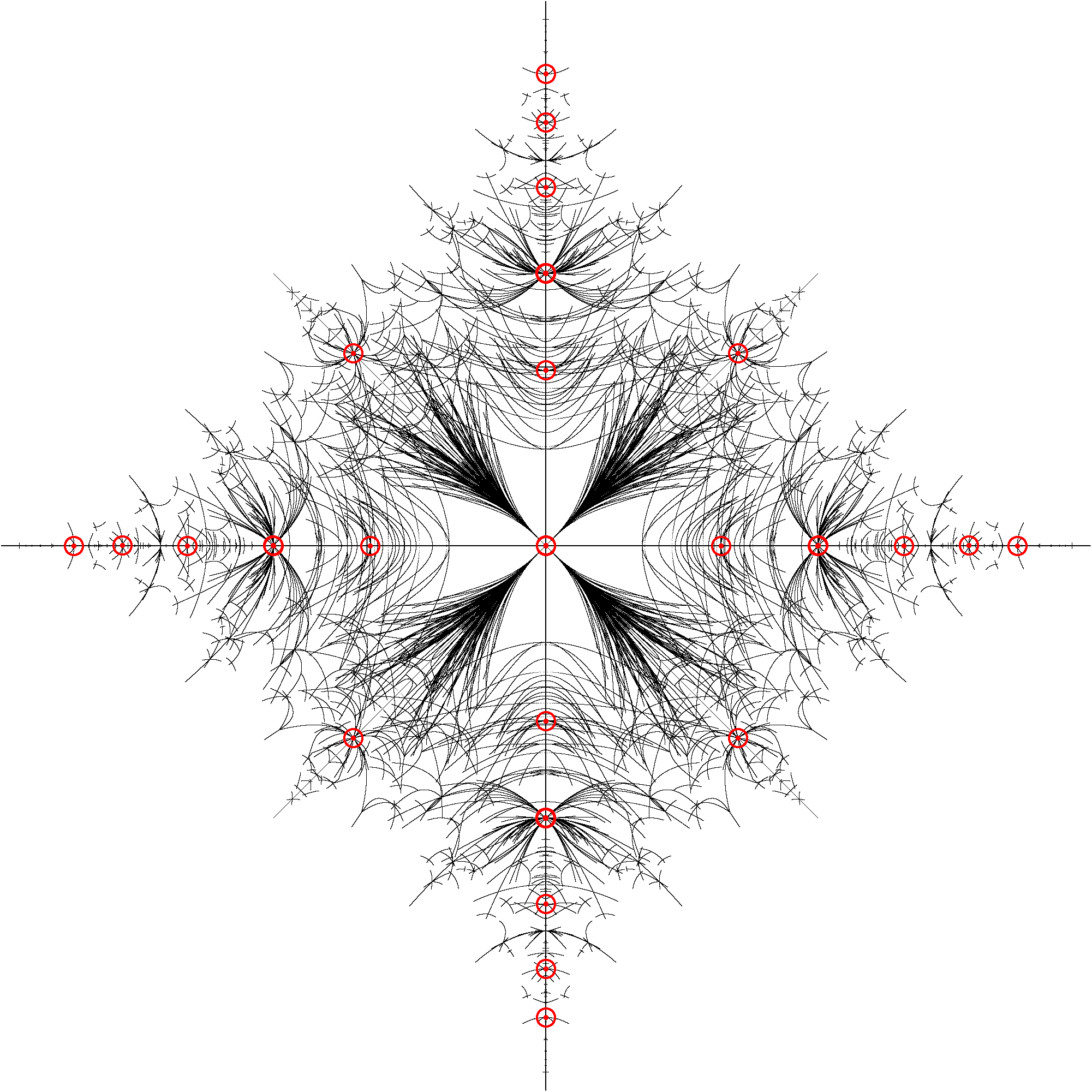}}\fi
\end{center}
\begin{figure}[h]
\caption{\footnotesize Here we see the inclusion $\sigma_5\ \subset\
\pi_{12}$, which holds by \eqref{eq:spefsinper} with $n=5$. (The points in $\sigma_5$ are indicated by circled dots.)}
\label{fig:pic5in12}
\end{figure}

Starting from $u_{0}=0$ and $u_{1}=1$, we will successively use
(\ref{eq:rec}) to compute $u_{i}$ for $i=2,3,...$ (an analogous
procedure is possible for $i=-1,-2,-3,...$) and see whether the
sequence remains bounded. Doing so, we get
\begin{eqnarray*}
u_{2}\!\!&=&\!\!\lambda,\qquad u_{3}\ =\ \lambda^2-c_2,\qquad
u_{4}\ =\ \lambda^3-(c_2+c_3)\lambda,\\
u_{5}\!\!&=&\!\!\lambda^4-(c_2+c_3+c_4)\lambda^2+c_2c_4,
\end{eqnarray*}
and so on.
Explicitly,
it is easy to check that, for $i\geq 3$, the solution of (\ref{eq:rec}) with initial
conditions $u_0=0$ and $u_1=1$ is given by the
characteristic polynomial
\[
u_{i} = \left|\begin{array}{rrrrr}\lambda&-1\\
-c_2&\lambda&\ddots\\
&\ddots&\ddots&-1\\
&&-c_{i-1}&\lambda
\end{array}\right|.
\]
Thus, for $i\geq 3$, $u_{i}$ is a polynomial of degree $i-1$ in
$\lambda$ with coefficients depending on $c_2, ..., c_{i-1}$. We
will aim to achieve that $u$ be a bounded sequence at least for
$|\lambda|<1$.

With this in mind we should try to keep the coefficients of these
polynomials small. Precisely, our strategy will be to try to choose
$c_1, c_{2}, ...\in\{\pm 1\}$ such that each $u_{i}$ is a polynomial
in $\lambda$ with coefficients in $\{-1,0,1\}$. The following table,
where we abbreviate $-1$ by ``$-$'', $+1$ by ``$+$'', and $0$ by a
space, suggests that this seems to be possible.

\begin{equation} \label{eq:table}
{\footnotesize
\begin{array}{|r|c|ccccccccccccccccc|}
\hline
&&&&&&&&&&&&&&&&&&\\[-2mm]
&&\multicolumn{4}{l}{j\ \to}&\multicolumn{12}{r}{\textrm{coefficients of $\lambda^{j-1}$ in the polynomial }u_i}&\\
i&c_{i}&1&2&3&4&5&6&7&8&9&10&11&12&13&14&15&16&\cdots\\
\hline
1&+&+&&&&&&&&&&&&&&&&\\
2&+& &+&&&&&&&&&&&&&&&\\
3&-&-& &+&&&&&&&&&&&&&&\\
4&-& & & &+&&&&&&&&&&&&&\\
5&+&-& &+& &+&&&&&&&&&&&&\\
6&-& &-& & & &+&&&&&&&&&&&\\
7&+&-& & & &+& &+&&&&&&&&&&\\
8&-& & & & & & & &+&&&&&&&&&\\
9&+&-& & & &+& &+& &+&&&&&&&&\\
10&+& &-& & & &+& & & &+&&&&&&&\\
11&-&+& &-& &-& & & &-& &+&&&&&&\\
12&+& & & &-& & & & & & & &+&&&&&\\
13&-&-& &+& & & & & &+& &-& &+&&&&\\
14&-& &-& & & & & & & &+& & & &+&&&\\
15&+&-& & & & & & & &+& & & &+& &+&&\\
16&-& & & & & & & & & & & & & & & &+&\\
\vdots&\vdots&&&&&&&&&\vdots&&&&&&&&\\ \hline
\end{array}
}
\end{equation}

For $i,j\in\N$, denote the coefficient of $\lambda^{j-1}$ in the
polynomial $u_{i}$ by $p_{i,j}$. Table
(\ref{eq:table}) shows the values $p_{i,j}$ for $i,j=1,...,16$, given the specific choices indicated on the left hand side of the table for the coefficients $c_i$. From
(\ref{eq:rec}) it follows that
\begin{equation} \label{eq:rec2}
p_{i+1,j}\ =\ p_{i,j-1}\ -\ c_{i}\,p_{i-1,j},
\end{equation}
for $i\in \N$ and $j=1,2,...,i+1$, where we have defined $p_{i,j}:=0$ if $j<1$, $i<1$,
or $j>i$.

Let us explore more systematically whether it is possible to choose the coefficients $c_i$ so as to ensure that all the coefficients $p_{i,j}\in \{-1,0,1\}$. Note first that, if this is possible, then if, for some $i,j$, one has that $p_{i,j-1}\ne 0$ and $p_{i-1,j}\ne 0$,
then $p_{i,j-1}$, $p_{i-1,j}\in\{-1,1\}$. Thus it follows from (\ref{eq:rec2}) that $p_{i+1,j}=0$, i.e.
\begin{equation} \label{eq:bi}
c_{i}\ =\ p_{i,j-1}/p_{i-1,j}\ =\ p_{i,j-1}\, p_{i-1,j},
\end{equation}
since otherwise $p_{i+1,j}\in\{-2,2\}$.  Illustrating this, look at
$p_{15,1}=-1$ and $p_{14,2}=-1$ in the above table. If we chose
$c_{15}=-1$, we would get from (\ref{eq:rec2}) that
$p_{16,2}=-2\not\in\{-1,0,1\}$, so it is necessary to choose
$c_{15}=1=p_{15,1}\, p_{14,2}$. Luckily, the same value
$c_{15}=1$ is required by the values of $p_{15,9}$ and $p_{14,10}$,
as well as by $p_{15,13}$ and $p_{14,14}$. We will prove that this
coincidence, i.e. that the right-hand side of (\ref{eq:bi}) is (if
non-zero) independent of $j$, is not a matter of fortune. As a
result we will show that the pattern of coefficients in table
(\ref{eq:table}) continues without end, only using values from
$\{-1,0,1\}$ for $p_{i,j}$ and from $\{\pm 1\}$ for $c_{i}$. To
prove this, we will make use of a particular self-similarity in the
triangular pattern of (\ref{eq:table}); more precisely, we will
show that the pattern of non-zero values of the coefficients $p_{i,j}$
forms a so-called infinite discrete Sierpinski triangle.

\begin{proposition} \label{prop:Sierpinski}
Define the sequence $c\in\{\pm 1\}^\Z$, for positive indices by
$c_1=1$ and by the requirement that
\[
c_{2i}\ =\ c_{2i-1}\,c_i\qquad\textrm{and}\qquad c_{2i+1}\ =\
-c_{2i},\qquad i=1,2,...\,,
\]
and for non-positive indices by
\[
c_{-i}= c_{i+1}, \quad i = 0,1,...\,.
\]
Further, given $\lambda\in \C$, define the sequence
$u=(u_{i})_{i\in\Z}$, by the requirement that
$$
u_{i+1}\ =\ \lambda u_{i} - c_i u_{i-1}, \quad i\in\Z,
$$
and by the initial conditions
$$
u_{0} = 0, \quad u_{1} = 1.
$$
Then, as a function of $\lambda$, for $i\in \Z$, $u_{i}$ is a
polynomial of degree $|i|-1$ with all its coefficients taking values
in the set $\{-1,0,1\}$.

In more detail, denoting, for $i,j\in \N$, the coefficient of $\lambda^{j-1}$ in the polynomial $u_i$ by $p_{i,j}$, the following statements hold.

\begin{itemize}
\item[(i)] $p_{i,j}=0$ for $j>i$, so that, for every $i\in\N$,
$$
u_{i}\ =\ \sum_{j=1}^{i} p_{i,j}\, \lambda^{j-1}.
$$
\item[(ii)] Defining, additionally, $p_{i,j}:= 0$ if $i,j\in \N\cup\{0\}$ and $i=0$ or $j=0$, it holds that $p_{1,1}=1$ and that
\begin{equation} \label{eqnnew}
p_{i+1,j}\ =\ p_{i,j-1}\ -\ c_{i}\,p_{i-1,j},
\end{equation}
for $i\in \N$ and $j=1,2,...,i+1$.
\item[(iii)] $p_{i,j}=0$ if $i+j$ is odd.
\item[(iv)] Writing the semi-infinite coefficient matrix $P=(p_{i,j})_{i,j\in\N}$ in block form as $P=(\mathbf{p}_{i,j})_{i,j\in\N}$ where, for $i,j\in\N$,
$$
\mathbf{p}_{i,j} := \left(\begin{array}{cc} p_{2i-1,2j-1}&p_{2i-1,2j}\\
p_{2i,2j-1}&p_{2i,2j}\end{array}\right),
$$
it holds, for $i\in\N$, that $\mathbf{p}_{i,j}=0$ for $j>i$ while,
for $j=1,...,i$,
\begin{equation} \label{eqneq2}
\mathbf{p}_{i,j} =
\left\{\begin{array}{cl}
p_{i,j}\left(\begin{array}{cc}1&0\\0&1\end{array}\right),&\textrm{if
$i+j$ is even},\\
c_{2i-1}\,p_{i-1,j}\left(\begin{array}{cc}1&0\\0&0\end{array}\right),&\textrm{if
$i+j$ is odd.}\end{array}\right.
\end{equation}
\item[(v)] $p_{i,j}\in\{-1,0,1\}$ for $i,j\in\N$.
\item[(vi)] Let $V := \{(0,0),(-1,-1),(1,-1)\}$ and let $S:=\left\{ (i,j)\in\N^2: p_{i,j} \mbox{ is non-zero}\right\}$.
Let $\Sigma:= 2^{\N^2}$ be the set of all subsets of $\N^2$, and
define $\bT:\Sigma\to\Sigma$ by
$$
\bT(\sigma)\ :=\ 2\sigma + V\ =\ \{2a+b: a\in \sigma, \, b\in V\},
\quad \mbox{ for }\sigma\in\Sigma.
$$
Then, where
$$
S_1\,:=\,\{(1,1)\}, \quad \mbox{ and }\quad S_{n+1}\,:=\,\bT(S_n),
\quad n\in\N,
$$
it holds that
$$
S\, =\, \bigcup_{n\in\N} S_n \quad \mbox{ and that } \quad S\,=\,
\bT(S).
$$
\item[(vii)] For $i\in \N\cup\{0\}$,
$$
u_{-i} = d_i\, u_{i},
$$
where, for $j\in\N\cup\{0\}$,
$$
d_{2j}:= (-1)^j c_{2j}, \quad d_{2j+1} := (-1)^{j+1}.
$$
\end{itemize}
\end{proposition}

~\vspace{-28mm}  

\noindent
\begin{tabular}{cp{113mm}}
\unitlength0.3mm 
\begin{picture}(100,100)(0,80)

\put(12,88){\line(0,-1){88}} \put(12,88){\line(1,0){88}}

\put(7,57){\makebox(0,0)[]{\footnotesize $i$}}
\put(5,25){\makebox(0,0)[]{\footnotesize $2i$}}
\put(38,95){\makebox(0,0)[]{\footnotesize $j$}}
\put(68,95){\makebox(0,0)[]{\footnotesize $2j$}}

\definecolor{grau1}{gray}{0.70}
\linethickness{2.7mm}
\put(33,56.5){\color{grau1}\line(1,0){9}}

\definecolor{grau2}{gray}{0.90}
\linethickness{6mm} \put(60,0){\color{grau2}\line(0,1){40}}
\linethickness{0.5pt}

\put(33,52){\framebox(9,9)[]{}} \put(38,56){\makebox(0,0)[]{$x$}}

\put(50,0){\framebox(20,40)[]{}} \put(50,20){\line(1,0){20}}
\put(56,34){\makebox(0,0)[]{$x$}} \put(66,25){\makebox(0,0)[]{$x$}}
\put(56,14){\makebox(0,0)[]{$y$}}

\multiput(14,57)(4,0){5}{\line(1,0){1}}
\multiput(14,25)(4,0){9}{\line(1,0){1}}
\multiput(37,85)(0,-4){6}{\line(0,-1){1}}
\multiput(66,86)(0,-4){12}{\line(0,-1){1}}

\put(44,56){\vector(1,-1){14}}
\end{picture} &
\begin{remark} \label{rem:ss}
Statements $(iv)$ and $(vi)$ reveal the self-similar nature of the
pattern (\ref{eq:table}). With respect to a scaling of the pattern
by the factor 2, an entry $p_{i,j}$, with $i+j$ even, replicates
three times: as $p_{2i-1,2j-1}$, $p_{2i,2j}$ and, multiplied by
$c_{2i+1}$, as $p_{2i+1,2j-1}$. So the ``volume'' of the pattern
\eqref{eq:table} triples under a scaling by 2, which is why (see
\cite{Dotyetal05}) its zeta dimension is $\log_2 3 \approx 1.585$ --
exactly the fractal (Hausdorff or box-counting) dimension of its
bounded version, the usual Sierpinski triangle or gasket
\cite{Falconer}.
\end{remark}
\end{tabular}

~ \\[-2mm]  

\noindent As an immediate consequence of Proposition
\ref{prop:Sierpinski} and formula \eqref{eq:spec} we get our main result.

\begin{theorem} \label{thm:D}
For the sequence $c\in\{\pm 1\}^\Z$ from Proposition
\ref{prop:Sierpinski}, it holds that the closed unit disk
$\overline\D:=\{z\in\C:|z|\le 1\}$ is contained in $\spec\, A^c$.
Consequently, for a pseudo-ergodic $b\in\{\pm 1\}^\Z$, one has
$\overline\D\subset\spec A^b$, so that $\spec A^b$ has dimension 2
and a positive Lebesgue measure.
\end{theorem}
\begin{proof} Let $\lambda\in\D:=\{z\in\C:|z|<1\}$, let $c$ be the sequence
from Proposition \ref{prop:Sierpinski} and $u:\Z\to\C$ the
corresponding eigenfunction from (\ref{eq:rec}). Then, for every
$i\in\N$,
\[
|u_{-i}|\ =\ |u_{i}|\ =\
\left|\sum_{j=1}^{i}p_{i,j}\lambda^{j-1}\right| \ \le\
\sum_{j=1}^{i}|p_{i,j}|\,|\lambda|^{j-1} \ \le\
\sum_{j=1}^\infty|\lambda|^{j-1}\ =\ \frac{1}{1-|\lambda|}
\]
since $p_{i,j}\in\{-1,0,1\}$ for all $i,j$, showing that
$u\in\ell^\infty(\Z)$, and, by our construction (\ref{eq:rec}),
$A^cu=\lambda u$. So $\D\subset\sppt A^c\subset\spec\, A^c$. Since
$\spec\, A^c$ is closed, it holds that
$\overline\D\subset\spec\,A^c$. The claim for a pseudo-ergodic $b$
now follows from $\spec A^c\subset\spec A^b$, by \eqref{eq:spec}.
Finally, from the monotonicity of (all notions of) dimension
\cite{Falconer}, it follows that $2=\dim(\overline\D)\le\dim(\spec
A^b)\le\dim(\R^2)=2$.
\end{proof}

\begin{proofof}{Proposition \ref{prop:Sierpinski}}
Statements $(i)$ and $(ii)$ are clear from the discussion preceding Proposition \ref{prop:Sierpinski}, and statement $(iii)$ then follows easily by induction. Thus $\mathbf{p}_{i,j}=0$ for $j>i$, and in every matrix $\mathbf{p}_{i,j}$ the off-diagonal entries are zero,
i.e. $p_{2i-1,2j}=0=p_{2i,2j-1}$ for all $i,j$.

We will now prove $(iv)$ by proving by induction that, for each $i\in\N$, \eqref{eqneq2} holds for $j=1,...,i$.
It is easy to check that \eqref{eqneq2} holds for $i=j=1$. Now suppose that, for some $k\in\N$, \eqref{eqneq2} holds for $i=1,...,k$, $j=1,..,i$.
We will show that this implies that \eqref{eqneq2} holds for $i=k+1$ and $j=1,...,k+1$.

We let $i=k+1$ and start with the case when $i+j$ is
even. By (\ref{eqnnew}) we have that
\begin{eqnarray}
\nonumber p_{2i-1,2j-1}&=&p_{2i-2,2j-2}\,-\,c_{2i-2}\, p_{2i-3,2j-1}\\
\label{eq:parti_11}&=&p_{2(i-1),2(j-1)}\,-\,c_{2i-2}\,
p_{2(i-1)-1,2j-1},
\end{eqnarray}
with $p_{2(i-1),2(j-1)}=p_{i-1,j-1}=0$ if $j=1$ and, by the inductive hypothesis (and since $i-1+j-1$ is even),
$p_{2(i-1),2(j-1)}=p_{i-1,j-1}$ if $j>1$.
Also, by the inductive hypothesis, $p_{2(i-1)-1,2j-1}=c_{2(i-1)-1}\,p_{i-2,j}$ since
$i-1+j$ is odd, while, from the definition of the sequence $c$, $c_{2i-2} = c_{2i-3}\,c_{i-1}$.
Inserting these results in (\ref{eq:parti_11}), we get that
\begin{eqnarray} \nonumber
p_{2i-1,2j-1}&=& p_{i-1,j-1}\, -\,
c_{2i-3}\,c_{i-1}\,c_{2i-3}\,p_{i-2,j}\\ \label{eqf}
&=& p_{i-1,j-1}\, -\,c_{i-1}\,p_{i-2,j}=
p_{i,j},
\end{eqnarray}
by \eqref{eqnnew}. We have observed already that $p_{2i-1,2j}=0=p_{2i,2j-1}$ for all $i,j$, so it remains to consider $p_{2i,2j}$. By \eqref{eqnnew}, \eqref{eqf}, and the inductive hypothesis which implies that $p_{2i-2,2j}=p_{2(i-1),2j}=0$ as $i-1+j$ is odd, we have that
\[
p_{2i,2j}\ =\ p_{2i-1,2j-1}\,-\,c_{2i-1}\,p_{2i-2,2j}\ =\ p_{i,j}.
\]

Now suppose $i+j$ is odd. Then, by \eqref{eqnnew} and the inductive hypothesis,
\begin{eqnarray*}
p_{2i-1,2j-1}&=&p_{2i-2,2j-2}\,-\,c_{2i-2}\,p_{2i-3,2j-1}\\
&=&0\,-\,c_{2i-3}\,c_{i-1}\,p_{i-1,j}\ =\ c_{2i-1}\,p_{i-1,j},
\end{eqnarray*}
since $c_{2i-1}=-c_{2i-2}=-c_{2i-3}\,c_{i-1}$. By \eqref{eqnnew} and the inductive hypothesis and noting that $i-1+j$ is even,
\begin{eqnarray*}
p_{2i,2j}&=&p_{2i-1,2j-1}\,-\,c_{2i-1}\,p_{2i-2,2j}\\
&=&c_{2i-1}\,p_{i-1,j}\,-\,c_{2i-1}\,p_{i-1,j}\ =\ 0.
\end{eqnarray*}

This completes the proof of $(iv)$, and $(v)$ follows from $(iv)$ by a simple induction argument.

To see that $(vi)$ is true, observe first that, from $(i)$, $(iii)$,
and $(iv)$ (and cf.\ Remark \ref{rem:ss}), it holds for
$i^\prime,j^\prime\in \N$ that $(i^\prime,j^\prime)\in S$ iff, for
some $i,j\in\N$ either $(i^\prime,j^\prime)=(2i,2j)$ and $(i,j)\in
S$; or $(i^\prime,j^\prime)= (2i-1,2j-1)$ and $(i,j)\in S$; or
$(i^\prime,j^\prime)= (2i+1,2j-1)$ and $(i,j)\in S$. From this it
follows that $S = \bT(S)$.

Define a metric $d$ on $\Sigma$ by
$$
d(\sigma,\tau)\ :=\ \sum_{(i,j)\in (\sigma \cup \tau)\setminus
(\sigma\cap \tau)} 2^{-i-j}, \qquad \sigma,\tau\in\Sigma.
$$
Then, since
$\big(\bT(\sigma)\cup\bT(\tau)\big)\setminus\big(\bT(\sigma)\cap\bT(\tau)\big)
\,\subset\, \bT\big((\sigma\cup\tau)\setminus(\sigma\cap\tau)\big)$
for all $\sigma,\tau\in\Sigma$,
\begin{eqnarray} \nonumber
d\left(\bT(\sigma),\bT(\tau)\right) &\leq & \sum_{(i,j)\in (\sigma \cup \tau)\setminus (\sigma\cap \tau)} \left(2^{-2i-2j}+2^{-(2i-1)-(2j-1)}+2^{-(2i+1)-(2j-1)}\right)\\
& = & \sum_{(i,j)\in (\sigma \cup \tau)\setminus (\sigma\cap \tau)}
2^{-i-j}\left(2^{-i-j}+2^{2-i-j}+2^{-i-j}\right)\ \leq\
\frac{3}{4}\, d(\sigma,\tau), \label{eq:cont}
\end{eqnarray}
if $(1,1)\not\in (\sigma \cup \tau)\setminus (\sigma\cap \tau)$. Let
$\Sigma_1:=\{\sigma\in \Sigma:(1,1)\in \sigma\}$. Then
$\bT(\Sigma_1)\subset \Sigma_1$ and, by \eqref{eq:cont}, $\bT$ is a
contraction mapping on  $\Sigma_1$. Thus, by the contraction mapping
theorem, $\bT$ has a unique fixed point in $\Sigma_1$, which is the
set $S$, and, if $\sigma_1\in \Sigma_1$ and $\sigma_{n+1}:=
\bT(\sigma_n)$, $n\in\N$, then $d(\sigma_n,S)\to 0$ as $n\to\infty$.
In particular, $d(S_n,S)\to 0$ as $n\to\infty$. Since also (by an
easy induction argument) $S_1\subset S_2 \subset ...$, it follows
that $S=\cup_{n\in\N} S_n$.

Define $v_{-i}$ for $i=0,1,...$ by $v_{-i}:= d_iu_i$, which implies that $v_0=0$, and set $v_1=1$. Then, since $u_{i}$ is defined uniquely for $i\leq 0$ by the requirement that it satisfy \eqref{eq:rec} for $i\leq 0$ with the initial conditions that $u_0=0$ and $u_1=1$, to show $(vii)$ it is enough to check that the sequence $v_i$ satisfies \eqref{eq:rec} for $i\leq 0$, i.e.\ that
$$
v_{-i+1} = \lambda v_{-i} - c_{-i} v_{-i-1}, \quad i=0,1,... \,.
$$
But $v_1-\lambda v_0 + c_0v_{-1} = 1 + c_0d_1u_1 =0$, so the equation holds for $i=0$, and, for $i\in\N$,
\begin{eqnarray*}
v_{-i+1} - \lambda v_{-i} + c_{-i} v_{-i-1} &=& d_{i-1} u_{i-1} -\lambda d_i u_i + c_{i+1} d_{i+1}u_{i+1}\\
& = & (d_{i-1}-c_ic_{i+1}d_{i+1})u_{i-1} - \lambda (d_i-c_{i+1}d_{i+1})u_i,
\end{eqnarray*}
since $u_{i+1} = \lambda u_i-c_iu_{i-1}$. Since $u_0=0$, the right hand side of this last equation is zero for $i\in\N$ provided that $d_i=c_{i+1}d_{i+1}$ for $i\in\N$. But this follows from the definitions of the sequences $c$ and $d$.
\end{proofof}

\begin{remark}
The standard infinite discrete Sierpinski triangle (e.g.\
\cite{LathropLutzSummers09}) is the set $\tilde S \subset \N^2$
defined by $\tilde S := \cup_{n\in\N} \tilde S_n$, where $\tilde
S_1:= \{(1,1)\}$ and the sets $\tilde S_n$, $n=2,3,...$, are defined
recursively by $\tilde S_{n+1} := 2\tilde S_n + \tilde V$, where
$\tilde V := \{(0,0),(-1,-1),(0,-1)\}$. One instance where $\tilde
S$ arises is as the pattern of odd coefficients in Pascal's
triangle: for $i\in \N$ and $j=1,...,i$ the coefficient of $x^{j-1}$
in $(1+x)^{i-1}$ is odd iff $(i,j)\in \tilde S$, so that the
discrete Sierpinski triangle is often referred to as Pascal's
triangle modulo 2 (e.g.\ \cite{Dotyetal05}). Proposition
\ref{prop:Sierpinski}$(vi)$ (cf.\ Remark \ref{rem:ss}) makes clear
that the pattern $S\subset \N^2$ of the non-zero coefficients in
table (\ref{eq:table}) is essentially that of the standard discrete
Sierpinski triangle $\tilde S$; indeed, the sets $\tilde S$ and $S$
are connected by a linear mapping:  $(i,j)\in \tilde S$ iff
$(2i-j,j)\in S$, for $i,j\in\N$.
\end{remark}

\begin{remark}
Note that the sequence $c$ from Proposition \ref{prop:Sierpinski} is
not pseudo-ergodic since, by $c_{2i+1}=-c_{2i}$, the patterns
``$+++$'' and ``$---$'' can never occur as consecutive entries in
the sequence $c$.
\end{remark}

\bigskip

Based on Theorems \ref{thm:tridiag} and \ref{thm:D} and the numerical results displayed in Figures \ref{fig:30pics1} and \ref{fig:30pics2}, we make the following conjecture.

\noindent {\bf Conjecture.} We conjecture that $\clos(\sigma_\infty) = \clos(\pi_\infty) = \spec A^b$, and that $\spec A^b$ is a
simply connected set which is the closure of its interior and which has a fractal boundary.

\bigskip


{\bf Acknowledgements.} We are grateful to Estelle Basor from the
American Institute of Mathematics for drawing our attention to this
beautiful operator class. Moreover, we would like to acknowledge: the
financial support of a Leverhulme Fellowship and a visiting
Fellowship of the Isaac Newton Institute Cambridge for the first
author; the invitation of the second and third author to the MPA
Workshop at the Isaac Newton Institute in July 2008; the financial
support of a Higher Education Strategic Scholarship for Frontier Research from the Thai Ministry of Higher Education to the second author; and the
Marie-Curie Grants MEIF-CT-2005-009758 and PERG02-GA-2007-224761 of
the EU to the third author.



\newpage  

\noindent {\bf Author's addresses:}\\[2mm]

\noindent Simon~N.~Chandler-Wilde\hfill
\href{mailto:s.n.chandler-wilde@reading.ac.uk}{{\tt
s.n.chandler-wilde@reading.ac.uk}}\\
\rule{3mm}{0pt} and Ratchanikorn Chonchaiya\hfill
\href{mailto:r.chonchaiya@reading.ac.uk}{{\tt
r.chonchaiya@reading.ac.uk}}\\[2mm]
Department of Mathematics\\
University of Reading\\
Reading, RG6 6AX\\
UK\\[5mm]

\noindent
Marko Lindner\hfill \href{mailto:marko.lindner@mathematik.tu-chemnitz.de}{{\tt marko.lindner@mathematik.tu-chemnitz.de}}\\
Fakult\"at Mathematik\hfill (corresponding author)\\
TU Chemnitz\\
D-09107 Chemnitz\\
GERMANY

\end{document}